\documentclass[11pt,a4paper]{article}
              \usepackage{a4wide}
\usepackage{color}
\usepackage[ruled,noend]{algorithm2e}
\usepackage
{epsfig,latexsym}

\newenvironment{proof}{\noindent{\bf Proof. }\ignorespaces}
{\hspace*{\fill}$\Box$\par\medskip}

\newtheorem{theorem}{Theorem}[section]

\newcommand{\remove}[1]{}

\begin{document}
\baselineskip=.22in

\newif\ifmarginnotes \marginnotestrue
\def\marginnotestyle
   {\scriptsize\itshape\raggedright}
\def\marginnote#1{
   \ifmarginnotes
     \marginpar{\marginnotestyle#1}%
   \fi
}

\title{\bf A Simple Algorithm for the \\ Constrained Sequence Problems}

\author{
{Francis Y.L. Chin$^{1,}$\footnote{This research was supported in part by Hong Kong RGC Grant}~,~ Alfredo De Santis$^2$, Anna Lisa Ferrara$^2$, N.L. Ho$^1$} and S.K. Kim$^3$ \\[.1in] 
{\small $^1$ Department of Computer Science and Information Systems}\\
{\small The University of Hong Kong,
 Pokfulam Road, Hong Kong}\\[.1in]
\vspace{5mm}
{\small E-mail: {\tt \{chin, nlho\}@csis.hku.hk}}\\
{\small $^2$ Dipartimento di Informatica ed Applicazioni}\\
{\small Universit\`{a} di Salerno,
 84081 Baronissi (SA), Italy}\\[.1in]
\vspace{5mm}
{\small E-mail: {\tt \{ads, ferrara\}@dia.unisa.it}}\\
{\small $^3$ Department of Computer Science and Engineering}\\
{\small Chung-Ang University, 221 Huksuk-dong Dongjak-ku} \\
{\small Seoul 156-756, Republic of Korea}\\[.1in]
\vspace{5mm}
{\small E-mail:  {\tt skkim@cau.ac.kr}}
}

\date{}

\maketitle
\begin{abstract}
\noindent In this paper we address the constrained longest common
  subsequence problem. Given two sequences $X$, $Y$ and a constrained
  sequence $P$, a sequence $Z$ is a constrained longest common subsequence
  for $X$ and $Y$ with respect to $P$ if $Z$ is the longest subsequence
  of $X$ and $Y$ such that $P$ is a subsequence of $Z$. \\
  Recently, Tsai~\cite{Tsai} proposed an $O(n^2 \cdot m^2 \cdot r)$ time
  algorithm  to solve this problem using dynamic programming technique,
  where $n$, $m$ and $r$ are the lengths of $X$, $Y$ and
  $P$, respectively. \\ In this paper, we present a simple algorithm
  to solve the constrained longest common subsequence problem in
  $O(n \cdot m \cdot r)$ time and show that the constrained longest
  common subsequence problem is equivalent to a special case of the
  constrained multiple sequence alignment problem which can also be solved
  with the same time complexity.
\end{abstract}
\footnotesize \noindent{ \bf \em keywords:}  Constrained Longest
Common Subsequence; Algorithms; Dynamic Programming; Sequence
Alignment. \normalsize

\newpage
\section{Introduction}

The longest common subsequence (LCS) problem has several
applications in many apparently unrelated fields, such as computer
science, mathematics, molecular biology, speech recognition, gas
chromatography. In molecular biology, LCS is an appropriate
measure of the similarity of biological sequences. Indeed, if we
wish to compare two strands of DNA or two protein sequences, we
may compute the LCS of them. The LCS problem on multiple sequences
is NP-hard~\cite{Mayer}. However, it may be solved in
polynomial-time for two sequences.

  Many algorithms have been designed using the dynamic programming
technique on this problem~\cite{Hirs, Masek, Wagner}. Tsai
addressed a variant of the LCS problem, {\em the constrained
longest common subsequence} (CLCS) problem~\cite{Tsai}. Given two
sequences $X$, $Y$ and a constrained sequence $P$, compute the
longest common subsequence $Z$ of $X$ and $Y$ such that $P$ is a
subsequence of $Z$. An $O(n^2 \cdot m^2 \cdot r)$ time algorithm
based on the dynamic programming technique was proposed for this
problem by Tsai~\cite{Tsai}, where $n$, $m$ and $r$ are the
lengths of $X$, $Y$ and $P$, respectively.

In this paper, we present a simple algorithm to solve the
CLCS problem in $O(n \cdot m \cdot r)$ time.
We further show that this CLCS problem is equivalent to a special case of
the constrained sequence alignment (CSA) problem which can be solved with the
same time complexity.

The rest of this paper is organized as follows. In
Section~\ref{introduction}, we define formally the CLCS problem
and characterize the structure in computing the CLCS. In
Section~\ref{dynamicprogram}, we present a simple dynamic
programming algorithm that computes the CLCS of two sequences with
respect to a constrained sequence. In Section~\ref{section:cmsa},
we show that the CLCS problem is equivalent to a special case of
the constrained sequence alignment problem~\cite{938103,cmsa02}.

\section{Characterization of the Constrained LCS Problem}
\label{introduction}

A sequence is a string of characters over a set of alphabets
$\Sigma$. A subsequence $Z$ of a sequence $X$ is obtained by
deleting some characters from $X$~(not necessarily contiguous); we
also say that $X$ \emph{contains} $Z$ if $Z$ is a subsequence of
$X$. Given two sequences $X$ and $Y$, $Z$ is a common subsequence
of $X$ and $Y$ if $Z$ is a subsequence of both $X$ and $Y$. $Z$ is
the \emph{longest common subsequence}~(LCS) of $X$ and $Y$ if $Z$
is the longest among all common subsequences of $X$ and $Y$. For
example, ``{\tt lm}'' and ``{\tt rm}'' are both the longest common
subsequences of ``{\tt problem}'' and ``{\tt algorithm}''. Let $P$
be a constrained sequence. We say that $Z$ is the
\emph{constrained LCS} of $X$ and $Y$ with respect to $P$ if $Z$
is the longest subsequence of $X$ and $Y$ and  $Z$ contains $P$
(i.e. $P$ is a subsequence of $Z$). For example, ``{\tt lm}'' is
the longest common subsequences of ``{\tt problem}'' and ``{\tt
algorithm}'' with respect to ``{\tt l}''.

Given a sequence $X = \langle x_1,x_2,\ldots,x_n \rangle$, where
character $x_i \in \Sigma$ for any $i=1, \ldots, n$, we denote the
$i$-th prefix of $X$ by $X_i = \langle x_1, x_2, \ldots, x_i
\rangle$ for any $i=1, \ldots, n$. In particular, $X_0$ denotes
the empty sequence. For example, if $X$=``{\tt algorithm}'' then
$X_3$=``{\tt alg}''.

 The following theorem
characterizes the structure of an optimal solution based on
optimal solutions to subproblems, for the constrained LCS problem.

\begin{theorem}\label{optimal_substructure}
If $Z = \langle z_1,z_2,\ldots,z_\ell \rangle$ is the constrained
LCS of $X = \langle x_1, x_2, \ldots, x_n \rangle$ and $Y =
\langle y_1,y_2,\ldots,y_m \rangle$ with respect to $P = \langle
p_1, p_2, \ldots, p_r \rangle$, the following conditions hold:

\begin{enumerate}
  \item If $x_n = y_m = p_r$ then $z_\ell = x_n = y_m = p_r$ and
    $Z_{\ell - 1}$ is the  constrained LCS of $X_{n-1}$ and
    $Y_{m-1}$ with respect to $P_{r-1}$.
  \item If $x_n = y_m$ and $x_n \ne p_r$
    then $z_\ell = x_n = y_m$ and $Z_{\ell-1}$ is the constrained LCS of $X_{n-1}$ and
    $Y_{m-1}$ with respect to $P$.
  \item If $x_n \ne y_m$ then $z_\ell \ne x_n$ implies that
    $Z$ is a constrained LCS of $X_{n-1}$ and $Y$ with respect to
    $P$.
  \item If $x_n \ne y_m$ then $z_\ell \ne y_m$ implies that
    $Z$ is a constrained LCS of $X$ and $Y_{m-1}$ with respect to
    $P$.
\end{enumerate}
\end{theorem}

\begin{proof}
As $Z$ is the constrained LCS of $X$ and $Y$ with respect to $P$,
 $x_n$, $y_m$ and $z_\ell$ are the last characters of $X$, $Y$ and $Z$,
respectively, we have $z_\ell = x_n = y_m$ if $x_n = y_m$.
 Assume by contradiction that $z_\ell \neq x_n$, we may
append $x_n=y_m$ to $Z$ to obtain a constrained common subsequence of $X$
and $Y$ of length $\ell+1$, contradicting the hypothesis that $Z$, of
length~$\ell$, is a constrained LCS of $X$ and $Y$ with respect to $P$.
 Therefore, if $x_n=y_m$ then $z_l=x_n=y_m$. This will be used in the
 proofs of $1$ and $2$. Now, we prove the four properties,

\begin{description}
\item[Proof of 1.] Since $x_n=y_m=p_r$, we have
$x_n=y_m=p_r=z_\ell$. The prefix $Z_{\ell-1}$ is a common
subsequence of $X_{n-1}$ and $Y_{m-1}$ with respect to $P_{r-1}$
of length $\ell-1$. Now, we show that $Z_{\ell-1}$ is a
constrained LCS of $X_{n-1}$ and $Y_{m-1}$ with respect to
$P_{r-1}$. Assume by contradiction that there exists a constrained
common subsequence $S$ of $X_{n-1}$ and $Y_{m-1}$ with respect to
$P_{r-1}$ whose length is greater than $\ell-1$. If we append
$x_n=y_m=p_r$ to $S$ we obtain a constrained common subsequence of
$X$ and $Y$ with respect to $P$ of length greater than $\ell$,
contradicting the hypothesis that $Z$ is a constrained LCS of $X$
and $Y$ with respect to $P$.

\item[Proof of 2.] Since $x_n=y_m$ and $x_n \neq p_r$,
 then $x_n=y_m=z_\ell$ and $z_\ell \ne p_r$.
 As $z_\ell \neq p_r$, the prefix $Z_{\ell-1}$ is a
common subsequence of $X_{n-1}$ and $Y_{m-1}$ with respect to $P$
of length $\ell-1$. Now, we show that $Z_{\ell-1}$ is a
constrained LCS of $X_{n-1}$ and $Y_{m-1}$ with respect to $P$.
Assume by contradiction that there exists a constrained common
subsequence $S$ of $X_{n-1}$ and $Y_{m-1}$ with respect to $P$
whose length is greater than $\ell-1$. If we append $x_n=y_m$ to
$S$, we obtain a constrained common subsequence of $X$ and $Y$ with
respect to $P$ of length greater than $\ell$, contradicting the
hypothesis that $Z$ is a constrained LCS of $X$ and $Y$ with
respect to $P$.

\item[Proof of 3.]
Since $z_\ell \neq x_ n$, $Z$ is a constrained
common subsequence of $X_{n-1}$ and $Y$ with respect to $P$. Now,
we show that $Z$ is a constrained LCS of $X_{n-1}$ and $Y$ with
respect to $P$. Assume by contradiction that there exists a
constrained common subsequence $S$ of $X_{n-1}$ and $Y$ with
respect to $P$ whose length is greater than $\ell$, then $S$ also
is a constrained common subsequence of $X$ and $Y$ with respect to
$P$ whose length is greater than $\ell$. This contradicts the assumption that
$Z$ is a constrained LCS of $X$ and $Y$ with respect to $P$.

\item[Proof of 4.]  The proof is similar to proof of 3.
\end{description}

\end{proof}

The next theorem shows a characterization of the constrained LCS
problem when no constrained common subsequence exists.
\begin{theorem}\label{optimal_substructure1}
If there is no constrained common subsequence of $X = \langle x_1,
x_2, \ldots, x_n \rangle$ and $Y = \langle y_1,y_2,\ldots,y_m
\rangle$ with respect to $P = \langle p_1, p_2, \ldots, p_r
\rangle$, the following conditions hold:

\begin{enumerate}
  \item If $x_n = y_m = p_r$ then there is no constrained common subsequence
  of $X_{n-1}$ and $Y_{m-1}$ with respect to $P_{r-1}$.
  \item There is no constrained common subsequence of the two sequences $X'$ and $Y'$ with
        respect to $P$, for each of the following three cases:
  \begin{itemize}
     \item $X'=X_{n-1}$ and $Y'=Y_{m-1}$;
     \item $X'=X_{n-1}$ and $Y'=Y$;
     \item $X'=X$ and $Y'=Y_{m-1}$.
  \end{itemize}

\end{enumerate}
\end{theorem}
\begin{description}
\item[Proof of 1.] Assume by contradiction that there is no
constrained common subsequence of $X$ and $Y$ with respect to $P$
but there exists a constrained common subsequence $Z$ of $X_{n-1}$
and $Y_{m-1}$ with respect to $P_{r-1}$. Since $x_n=y_m=p_r$ then
the concatenation of $x_n$ to $Z$ is a constrained common
subsequence of $X$ and $Y$ with respect to $P$. Contradiction.

\item[Proof of 2.] Assume by contradiction that there is no
constrained common subsequence of $X$ and $Y$ with respect to $P$
but there exists a constrained common subsequence $Z$ of $X'$ and
$Y'$ with respect to $P$. This is a contradiction because $Z$ is
also a constrained common subsequence of $X$ and $Y$ with respect
to $P$.
\end{description}

\section{A Simple Algorithm}\label{dynamicprogram}
Given two sequences $X$, $Y$  and a constrained sequence $P$,
whose lengths are $n$, $m$ and $r$, respectively, we define
$L(i,j,k)$ as the  constrained LCS length of $X_i$ and $Y_j$ with
respect to $P_k$, for any $0 \le i \le n$, $0 \le j \le m$ and $0
\le k \le r$. In particular, $L(n,m,r)$ gives the length of the
constrained LCS
 of $X$ and $Y$ with respect to $P$.
We design an algorithm that computes the CLCS of $X$ and $Y$ with
respect to $P$ in $O(n \cdot m \cdot r)$ time.

If either $i<k$ or $j<k$, there is no constrained common
subsequence for $X_i$ and $Y_j$ with respect to $P_k$. We
represent this condition by denoting $L(i,j,k)= - \infty$,
where~$\infty$ represents a large number, greater than the maximum
value of~$n$ and~$m$.
 If $i=0$ or $j=0$ and $k=0$, the CLCS for $X_i$ and $Y_j$ with
respect to $P_0$ has length $0$. The characterization of the
structure of a solution for the CLCS problem based on  solutions
to subproblems shown in Section \ref{introduction}, yields the
following recursive relation, for any $0 \le i \le n$, $0 \le j
\le m$ and $0 \le k \le r$,
\bigskip
\begin{equation}
   L(i,j,k) =
    \left\{ \begin{array}{llrl}
        1 + L(i-1,j-1,k-1) & $ if $ i,j,k > 0 $ and $ x_i = y_j = p_k \\

        1 + L(i-1,j-1,k)   & $ if $ i,j > 0, x_i = y_j $ and
           $ (k=0 $ or $ x_i \ne p_k)\\

        \max( L(i-1,j,k), L(i,j-1,k) )   & $ if $ i,j > 0 $ and $
                                           x_i \ne y_j\\
      \end{array} \right.
    \label{recursionequation}
\end{equation}
    with boundary conditions, $L(i,0,0) = L(0,j,0) = 0$ and
    $L(0,j,k) = L(i,0,k) = -\infty$, for $i=0,\ldots,n$,
    $j=0, \ldots,m$ and $k=1,\ldots,r$.

\noindent This is a generalization of the recurrence formula that
computes the length of an LCS between two sequences
~\cite{leiserson}, indeed, if $k=0$, it holds that,
\bigskip
\begin{equation}
  L(i,j,0) =
    \left\{ \begin{array}{llrl}
        1 + L(i-1,j-1,0)   & $ if $ i,j > 0, x_i = y_j \\
        \max( L(i-1,j,0), L(i,j-1,0))        & $ if $ i,j > 0, x_i \ne y_j\\
        0         & $ if $ i=0 $ or $ j = 0\\
    \end{array} \right.
    \label{lcs_recurrence}
\end{equation}
\bigskip

\begin{paragraph}{Constructing the Constrained LCS}
\label{clcs_construction} The CLCS of $X$ and $Y$ with respect to
$P$ can be constructed by backtracking through the computation
path from $L(n,m,r)$ to $L(0,0,0)$. Let $Z$, the initial CLCS, be
an empty sequence. If the value of $L(i,j,k)$ is computed from
$L(i-1,j-1,k)$ or $L(i-1,j-1,k-1)$, prepend the character
$x_i$($=y_j$) to $Z$. Repeat backtracking until reaching
$L(0,0,0)$, and $Z$ is the
 CLCS of $X$ and $Y$ with respect to $P$. Recovering the computation
path of the  CLCS takes at most $O(n+m+r)$ steps.
\end{paragraph}

Thus, computing and constructing the CLCS takes
$O(n \cdot m \cdot r)$ time and space.


\section{CLCS and Constrained Sequence Alignment}
\label{section:cmsa} In this section, we show that the CLCS
problem is in fact a special case of the \emph{constrained
sequence alignment} (CSA) problem~\cite{938103,cmsa02}.

Let $X=\langle x_1,x_2,\cdots,x_n \rangle$ and
$Y=\langle y_1,y_2,\cdots,y_m \rangle$ be two sequences
over $\Sigma$, with lengths $n$ and $m$, respectively.
We define the \emph{sequence alignment} of $X$ and
$Y$ as two equal-length sequences
$X'=\langle x'_1,x'_2,\cdots,x'_{n'} \rangle$ and
$Y'=\langle y'_1,y'_2,\cdots,y'_{n'} \rangle$ such that $|X'| = |Y'| = n'$,
where $n' \ge n,m$,
and removing all space characters ``{\tt -}'' from $X'$ and $Y'$ gives
$X$ and $Y$, respectively, with the assumption that no
$x'_{i} = y'_{i} = $~``{\tt -}'' for any~$1 \le i \le n'$.
For a given distance function $\delta(x',y')$ which
measures the \emph{mutation} distance between two characters, where
$x',y' \in \Sigma \cup \{${\tt -}$\}$, the
 \emph{alignment score} of two length-$n'$ sequences
$X'$ and $Y'$ is defined as
$\sum_{1 \le i \le n'}{\delta(x'_i,y'_i)}$.

In the \emph{constrained sequence alignment (CSA) problem},
we are given, in addition to the  inputs of the sequence alignment problem,
 a constrained sequence $P=\langle p_1,p_2, \cdots,p_k \rangle$,
 where  $P$ is a common subsequence of $X$ and $Y$.
 A solution of the CSA problem is
\Big[$X' \atop Y'$\Big], an alignment of $X$ and $Y$,
such that when $X'$ is placed on top of $Y'$,
 each character in $P$ appears in an entire column of the alignment
and in the same order as $P$, i.e. there exists a list of integers
$\langle c_1,c_2,\ldots,c_r \rangle$ where
$ 1 \le c_1 < c_2 < \ldots  < c_r \le n'$,
 and for all $1 \le k \le r$, we have $x'_{c_k} = y'_{c_k} = p_k$.
The CSA problem is to find $X'$ and $Y'$ with minimum alignment score
when given two sequences $X$, $Y$, a constrained sequence $P$ and a
distance function $\delta(x',y')$.

The CSA problem can also be solved in $O(n \cdot m \cdot r)$ time and space\cite{938103}. Next, we show that the CLCS problem is equivalent  to the
CSA problem  of $X$ and $Y$ with respect to $P$, using the  distance function
 $\delta(x',y')$, where $x',y' \in \Sigma \cup \{${\tt -}$\}$,
\begin{equation}
 \label{distance}
 \delta(x',y') = \left\{
  \begin{array}{lll}
    -1 & $ if $ x' = y'   &$ {\small \it (match)} $\\
    0  & $ if $ x' \ne y' &$ {\small \it (insertion, deletion or replacement)}$\\
  \end{array}
 \right.
\end{equation}

The distance function $\delta(x',y')$ in Equation~(\ref{distance})
favors matching characters, and does not penalize mismatching characters or
 insertion of spaces.
Therefore, when the CSA alignment score is $-s$,  there are $s$ matchings
in $X$ and $Y$ with respect to $P$.

\begin{theorem}
Given two sequences $X$, $Y$ and a constrained sequence $P$, the CLCS of
problem is equivalent to the $CSA$ problem when the distance
function $\delta(x',y')$ given in Equation~(\ref{distance}) is used.
\end{theorem}
\begin{proof}
 Let $[ {X' \atop Y'} ]$ be the CSA solution with the minimum
 alignment score, $n' = |X'| = |Y'|$.
 By the definition of CSA with the distance function $\delta(x',y')$,
 $[ {X' \atop Y'} ]$ has the minimum alignment score only if
 $X'$ and $Y'$ have the most number of matches~(i.e. $x'_i = y'_i$)
 and every character in $P$ appears as a
 column in $[ {X' \atop Y'} ]$.
 Let $Z'$ be the subsequence of $X'$ and $Y'$, containing \emph{only}
 the matching characters in $X'$ and $Y'$.
 Obviously, $Z'$ is a common subsequence of $X'$ and $Y'$ containing $P$.
 Thus, $Z'$ is the CLCS of $X$ and $Y$ with respect to $P$.

   Let $Z$ be a CLCS solution of $X$ and $Y$ with respect to $P$ and
   $|Z|=\ell$. By definition of CLCS, $P$ is a subsequence
   of $Z$ and $Z$ is a common subsequence of both $X$ and $Y$.
   To obtain an optimal solution for the CSA problem of $X$ and $Y$
   with respect to $P$,
   we can construct $X'$ and $Y'$ by inserting spaces into $X$ and $Y$
   respectively, such that every character in $Z$ appears in the same
   position in $X'$ and $Y'$.
   Using the distance function $\delta(x',y')$,
   the alignment score of $[ {X' \atop Y'} ]$ is~$-\ell$,
   i.e., the alignment $[ {X' \atop Y'} ]$ has $\ell$ matching columns.
   Since $P$ is a subsequence of $Z$, there is a column matching
   in $[ {X' \atop Y'} ]$ for every character in $P$.
   As $Z$ is a CLCS of $X$ and $Y$ with respect to $P$, the optimal CSA
   solution for $X$ and $Y$ with respect to $P$ has at
   most~$\ell$~matches, i.e. with the minimum alignment score $-\ell$.
   Hence, $[ {X' \atop Y'} ]$ is the optimal
   CSA solution for $X$ and $Y$ with respect to $P$.
\end{proof}


\section{Conclusions}

In this paper, we have addressed the constrained longest common
subsequence problem proposed by Tsai~\cite{Tsai}. An $O(n^2\cdot
m^2 \cdot r)$ time algorithm based on the dynamic programming
technique was proposed to compute a CLCS for $X$ and $Y$ with
respect to $P$, where $n$, $m$ and $r$ are the lengths of $X$, $Y$
and $P$, respectively. We have described a simple $O(n\cdot m
\cdot r)$ time algorithm to solve this problem and have also
showed that the CLCS problem is indeed a special case of the CSA
problem.

It is not difficult to show that this problem can also be solved
in~$O( \min(n,m) \cdot r)$ space based on  Hirschberg's
Algorithm~\cite{Hirs}. We define $L^r(i,j,k)$ as the length of
CLCS of the suffices of $X$, $Y$ and $P$ starting from the $i$-th,
$j$-th and $k$-th positions, respectively. WLOG, assume $n \ge m$,
then $Z$, the solution for the CLCS problem, can be defined as
$\max_{0 \le j \le m, 0 \le k \le r }
  \Big( L(\frac{n}{2},j,k) + L^r(\frac{n}{2}+1, j+1, k+1) \Big)$ which
can be found in $O(n \cdot m \cdot r)$ time~(say $K n  m r$ time, 
where $K$ is a constant) and $O(m \cdot r)$ space.
Assume the maximum value occurs when $j = j'$ and $k = k'$, then we
can further solve the two  subproblems $L(\frac{n}{2},j',k')$ and
$L^r(\frac{n}{2}+1, j'+1, k'+1)$ in $\frac{1}{2}K  m  n  r$ time
and $O(m \cdot r)$ space.
Continuing this recursively, we can solve the CLCS problem
in the same  $O(n \cdot m \cdot r)$ time and $O(m \cdot r)$
space complexities.


\begin{thebibliography}{99}

{\small
\bibitem{938103}
Francis Y.~L. Chin, N.~L. Ho, T.~W. Lam, Prudence W.~H. Wong, and M.~Y. Chan.
\newblock {\em Efficient constrained multiple sequence alignment with
  performance guarantee}.
\newblock In  Proceedings of the IEEE Computer Society Conference on
  Bioinformatics. IEEE Computer Society, 2003, 337--346.

\bibitem{leiserson}
T.H. Cormen, C.E. Leiserson and R.L. Rivest {\em Introduction to
Algorithms}, MIT Press, Cambridge, MA, 1990.

\bibitem{Hirs}
D.S. Hirschberg, {\em Algorithms for the Longest Common
Subsequence Problem}, J. ACM 24,  1977, 664--675.

\bibitem{Mayer}
D. Mayer, {\em The Complexity of Some Problems on Subsequences and
Supersequences}, J. ACM 25,  1978, 322--336.

\bibitem{Masek}
W.J. Masek and M.S. Paterson, {\em A Faster Algorithm Computing
String Edit Distances}, J. Comput. System Sci. 20, 1980, 18--31.

\bibitem{cmsa02}
C.Y. Tang, C.L. Lu, M.D.T. Chang, Y.T. Tsai, Y.J. Sun, K.M. Chao, J.M. Chang, Y.H. Chiou, C.M. Wu, H.T. Chang and W.I. Chou,
{\em Constrained Multiple Sequence Alignment Tool Development and Its Application to RNase Family Alignment}, Journal of Bioinformatics and Computational Biology, 1, 2003, 267--287.

\bibitem{Tsai}
Y.T. Tsai, {\em The Constrained Longest Common Subsequence
Problem}, Information Processing Letters 88, 2003, 173--176.

\bibitem{Wagner}
R.A. Wagner and M.J. Fischer, {\em The String-to-String correction
Problem}, J. ACM 21, 1974, 168--173. }
\end{thebibliography}
\end{document}